\documentclass[submission,copyright,creativecommons]{eptcs}
\providecommand{\event}{SOS 2007} 
\usepackage{amssymb}
\usepackage{amsmath}
\usepackage{stmaryrd}
\usepackage{xcolor}
\usepackage{iftex}
\usepackage{amsthm}
\theoremstyle{plain}
\newtheorem{theorem}{Theorem}
    
\newtheorem{proposition}[theorem]{Proposition}
\theoremstyle{definition}
\newtheorem{definition}{Definition}
\newtheorem{example}{Example}
\newcommand{\truthset}[3]{\llbracket{#1}\rrbracket_{#2}^{#3}}
\newcommand{\truthsett}[3]{[{#1}]_{#2}^{#3}}

\ifpdf
  \usepackage{underscore}         
  \usepackage[T1]{fontenc}        
\else
  \usepackage{breakurl}           
\fi

\title{On {\L}ukasiewicz Topological Semantics}
\author{Muhammad Afaq Khan \quad Manuel A. Martins \quad Alexandre Madeira \quad Babar Zaman
\institute{CIDMA \& Department of Mathematics, University of Aveiro
Aveiro, Portugal}
\email{\{mafaqkhan,martins,madeira,babar.zaman\}@ua.pt}
}
\def\titlerunning{On {\L}ukasiewicz Topological Semantics}
\def\authorrunning{Khan et al.}
\begin{document}
\maketitle

\begin{abstract}
In this paper, we propose fuzzy topological semantics for {\L}ukasiewicz valued modal logic, based on the standard {\L}ukasiewicz algebra. Within this setting, formulas are interpreted as fuzzy sets of truth values, while the modal operators are realized through fuzzy interior and closure operators. We define a Kripke-induced fuzzy Alexandroff topology and establish a precise correspondence between {\L}ukasiewicz fuzzy Kripke semantics and its topological counterpart. The proposed topological semantics validates the $S4$ modal axioms that include reflexivity, transitivity, monotonicity, normality, and the necessitation of truth.
\end{abstract}

\section{Introduction}
Modal logic studies necessity, possibility, and related intensional notions through the operators $\Box$ and $\Diamond$, and Kripke semantics interprets these operators by means of an accessibility relation between possible worlds (see \cite{blackburn2001modal, kripke1963semantical}). Relational Kripke semantics serves as the conventional basis for interpreting modal logic. However, when this framework is carried over to the many-valued setting, a number of technical obstacles arise. Within many-valued modal semantics, the truth value of a formula is no longer confined to binary options. Instead, it is drawn from an algebraic structure such as the real unit interval $[0,1]$, and accessibility relations may also be graded. This line of research goes back to many-valued modal logics (see \cite{fitting1991many}) and has been developed through fuzzy Kripke models, modal systems over finite lattices, and modal expansions of Gödel, Product, and {\L}ukasiewicz logics, we refer the interested reader to \cite{bou2009characterizing,bou2011minimum,jain2020fuzzy,vidal2017modal}. Note, however, that this semantic adjustment affects several aspects, including well-known properties of the modalities, such as duality and the normality axiom, which do not hold automatically over arbitrary many-valued Kripke frames, as discussed by Amanda Vidal in \cite{vidal2015modal}. To restore soundness and completeness, one is therefore led to impose further conditions on the semantics, for instance crispness or idempotence of the accessibility relation (see \cite{vidal2021transitive, hansoul2013extending}).\\
Fuzzy topological semantics discussed in \cite{khan2025fuzzytopologicalsemantics} presents a natural alternative to analyze modality in the many-valued context. The classical topological semantics interpret the operators $\Box$ and $\Diamond$ as the interior and closure operators of a topology (see \cite{mckinsey1944algebra, goldblatt2014topoi}). This interpretation lends modal logic a spatial, neighbourhood-oriented meaning. The classical account nevertheless remains crisp, since every formula is assigned one of only two values at each point. It has no means of recording truth value of a formula holds to degree, say 0.7 or 0.4, at a point. Fuzzy topology removes this restriction by admitting fuzzy open sets in place of ordinary ones, so that both membership and openness become matters of degree. For a comprehensive introduction to fuzzy topology, we direct the reader to \cite{chang1968fuzzy, goguen1973fuzzy}. The result is a semantics in which spatial approximation and many-valued truth can be expressed together.\\
In our previous work \cite{khan2025fuzzytopologicalsemantics}, we developed a fuzzy topological semantics for Gödel modal logic and established its correspondence with Gödel fuzzy Kripke semantics. Our work showed that fuzzy topological semantics can provide a meaningful spatial interpretation of many-valued modal reasoning. It also suggested an important observation that the fuzzy topological structure must reflect the algebraic behaviour of the underlying many-valued logic. In the Gödel setting, conjunction is given by the idempotent lattice meet $\wedge$, so the usual closure of fuzzy open sets under finite meets is naturally compatible with the logic. This compatibility does not directly carry over to logics with non-idempotent operations.\\
Guided by this observation, the present paper develops a fuzzy topological semantics for modal logic over the standard Łukasiewicz algebra. Łukasiewicz logic differs from Gödel logic in an essential aspect, since it relies on non-idempotent operations. The most prominent of these is the strong conjunction $\odot$, which is accompanied by the bounded sum $\oplus$, an implication, and the involutive negation $\neg a=1-a$ (see \cite{hajek2001metamathematics, cignoli2013algebraic}). A topological semantics for the Łukasiewicz modal language therefore cannot rest on lattice meet and arbitrary join alone. It must also contain strong conjunction. Rather than adjusting the accessibility relation in order to recover ordinary modal behaviour, we modify the topological structure itself. Formulas in the proposed framework are read as fuzzy subsets of a space, namely as functions $[0,1]^X$. The necessity operator $\Box$ is given by a Łukasiewicz fuzzy interior operator. Because Łukasiewicz negation is involutive, the possibility operator can be defined dually through $ \Diamond\varphi:=\neg\Box\neg\varphi$.\\
The main contribution of this paper is the construction of a sound fuzzy topological semantics for {\L}ukasiewicz modal logic. We establish a correspondence between {\L}ukasiewicz fuzzy Kripke semantics and the proposed fuzzy topological semantics by showing that each suitable {\L}ukasiewicz fuzzy Kripke frame induces an associated fuzzy topological space, and that the interpretation of modal formulas is preserved under this construction. We also prove the validity of the standard $S4$ principles, including the interior principles corresponding to reflexivity and transitivity, together with normality. Thus, the paper extends fuzzy topological semantics from the Gödel case to the {\L}ukasiewicz case in a way that is both algebraically faithful and topologically meaningful.\\

\noindent\textbf{Outline of the paper.} The paper is organized as follows. In Section \ref{section} we recall the topological semantics, the basic operations of the standard Łukasiewicz algebra and the modal language, together with fuzzy topological spaces and their interior and closure operators. In Section \ref{main section} we present Łukasiewicz fuzzy topological model and explain how formulas are interpreted under this construction. Next, the correspondence with fuzzy Kripke semantics is established. The final part of the paper proves the core interior and closure properties, confirms the validity of the S4-style modal principles, offers a finite counter example, and discuss possible extensions.

\section{Preliminaries}\label{section}
In this section, we recall some of the foundational concepts and notions that are necessary to understand the core theme of the paper. We assume that the reader is familiar with the basic notions of modal logic and classical Kripke semantics. Therefore, we do not recall these topics in full detail. This includes fundamental definitions of related to Modal language, topological models, fuzzy modal logic, {\L}ukasiewicz and {\L}ukasiewicz fuzzy operations, that will lead us to construct the definition of {\L}ukasiewicz fuzzy topological models.

\textbf{Topological semantics for Modal Logic.}
Topological semantics offers a natural and flexible interpretation of modal logic. Unlike Kripke semantics, where modal operators are interpreted through an accessibility relation between worlds, topological semantics interprets $\Box$ as an interior operator and $\Diamond$ as a closure operator (see \cite{aiello2007handbook, van2007modal}). This gives modal formulas a spatial meaning: necessity corresponds to truth in an open neighbourhood, while possibility corresponds to truth in the closure or approximation of a region. This approach is useful because many modal axioms can be explained directly through topological properties. In particular, the axioms of $S4$ arise naturally from the basic behaviour of the interior operator. Moreover, topological semantics is well suited for settings where reasoning depends on neighbourhoods, approximation, continuity, or regions rather than only on relations between individual worlds. For fuzzy and many-valued modal logics, this becomes even more meaningful, since formulas can be interpreted as fuzzy subsets and modal operators can be studied using fuzzy interior and closure operations.
\begin{definition}[Modal language {\cite{aiello2007handbook}}]
Let $Prop$ be a countable set of propositional letters. The basic modal language $\mathcal{L}$ is generated by the following grammar:
$$
\varphi,\psi ::= p \mid \neg\varphi \mid \varphi\wedge\psi
\mid \varphi\vee\psi \mid \varphi\to\psi
\mid \Box\varphi \mid \Diamond\varphi,
\qquad p\in Prop.
$$
Here, $\Box$ and $\Diamond$ denote the necessity and possibility modal operators, respectively.
\end{definition}
\begin{definition}[Topological model {\cite{aiello2007handbook}}]
A \emph{topological model}, or simply a \emph{topo-model}, is a triple
$M=(X,\tau,V)$, where $(X,\tau)$ is a topological space and
$V:Prop\longrightarrow \mathcal{P}(X)$ is a valuation function assigning to each propositional letter a subset of $X$.
\end{definition}
\begin{definition}[Topological satisfaction relation]
Let $M=(X,\tau,V)$ be a topological model. The satisfaction relation $M,x\models\varphi$, for $x\in X$, is defined recursively as follows:
\begin{itemize}
\item $M,x\models p$ iff $x\in V(p)$, for $p\in Prop$;
\item $M,x\models \neg\varphi$ iff $M,x\not\models \varphi$;
\item $M,x\models \varphi\wedge\psi$ iff $M,x\models\varphi$ and $M,x\models\psi$;
\item $M,x\models \varphi\vee\psi$ iff $M,x\models\varphi$ or $M,x\models\psi$;
\item $M,x\models \Box\varphi$ iff
$x\in \operatorname{Int}({y\in X:M,y\models\varphi})$;
\item $M,x\models \Diamond\varphi$ iff
$x\in \operatorname{Cl}({y\in X:M,y\models\varphi}).$
\end{itemize}
\end{definition}

A formula $\varphi$ is valid in a topological model $M$ if for all $x \in X$, the condition $M,x \vDash \varphi$ holds. Moreover, it is valid in topological space $(X,\tau)$ if it is valid
in any topological model $M$ over $(X,\tau)$.

\medskip

\noindent
Note: Every topological model satisfies the axioms of S4 modal logic \cite{van2007modal}, i.e. the minimum collection of modal formulas that include all classical tautologies. The equivalent axiomatization of S4 relevant to our study is defined by the axioms: $(N)\ \Box \top$, $(T)\ \Box p \to p$, $(R)\ \Box(p \wedge q) \leftrightarrow \Box p \wedge \Box q$
and $(4)\ \Box p \to \Box\Box p$ and is closed under modus ponens, substitution and monotonicity (from $\varphi \to \psi$ derive $\Box\varphi \to \Box\psi$).

\textbf{Łukasiewicz Modal Logic: Kripke semantics.}
Fuzzy modal logics extend modal reasoning by allowing truth values in the interval $[0,1]$ rather than only the classical values $0$ and $1$ (see \cite{fitting1991many,vidal2017modal, madeira2016exercise}). 
The logical behaviour of these systems depends on the chosen $t$-norm $T$, which interprets conjunction, and on its associated residuum $\Rightarrow$, which interprets implication for detailed study reader is referred to \cite{hajek2001metamathematics, esteva2001monoidal}. A $t$-norm is a commutative, associative, and monotone operation $T:[0,1]^2\to[0,1]$ such that $T(1,x)=x$ and $T(0,x)=0$. Important examples are the \L{}ukasiewicz $t$-norm $T(x,y)=\max(0,x+y-1)$, the G\"odel $t$-norm $T(x,y)=\min(x,y)$, and the product $t$-norm $T(x,y)=xy$. For a continuous $t$-norm, the corresponding residuum is defined by
\[
x\Rightarrow y=\max\{z\in[0,1]\mid T(x,z)\leq y\}.
\]
Most fuzzy modal logics have semantics based on continuous $t$-norms and its residuum. In what follows, we consider a generic continuous $t$-norm $T$ and its residuum $\Rightarrow$. We focus on the fuzzy modal logic over {\L}ukasiewicz algebra.

\begin{definition}[{\L}ukasiewicz algebraic operations \cite{hajek2001metamathematics}]
Let $[0,1]$ be the real unit interval. The standard {\L}ukasiewicz operations are defined, for all $a, b \in [0,1]$, as follows:
\begin{align*}
    a\oplus b=\min(1, a+b), \qquad& a\odot b=\max(0, a+b-1),\\
    a\Rightarrow b=\min(1, 1-a+b), \qquad& \neg a=1-a.
\end{align*}
Here $\oplus$ is the {\L}ukasiewicz bounded sum, $\odot$ is the {\L}ukasiewicz strong conjunction, $\Rightarrow$ is the {\L}ukasiewicz implication and $\neg$ is the {\L}ukasiewicz negation. The residuum
condition is $a\odot c \leq b \quad \text{ iff } \quad c \leq a\Rightarrow b$ (see \cite{klement2013triangular}). The negation is involutive (see \cite{chang1958algebraic}), i.e. $\neg\neg a=a$.
\end{definition}
\begin{definition}\label{language}
Let $Prop$ be a set of atomic propositions, the set of modal formulas $Form$ is given by (see \cite{vidal2015modal}):
\begin{center}
     $\varphi,\psi ::= p \mid  \bot \mid \varphi\wedge\psi \mid \varphi\odot\psi \mid  \varphi\rightarrow\psi \mid \Diamond\varphi\mid \Box \varphi$, where $p\in Prop$.
\end{center}
The remaining connectives are defined by
    \begin{align*}
        \neg\varphi:=\varphi &\to \bot, \quad \top:=\bot\to\bot,\\
        \varphi\leftrightarrow\psi&:=(\varphi\to\psi)\land(\psi\to\varphi).
    \end{align*}
\end{definition}
In this language, we distinguish between two conjunctions. The connective $\wedge$ denotes the lattice conjunction, interpreted by the minimum operation, while $\odot$ denotes the strong, or monoidal, conjunction of \L{}ukasiewicz logic, interpreted by $a\odot b=\max(0,a+b-1)$. This distinction is standard in t-norm based many-valued logics and also appears in Vidal's formulation of many-valued modal languages, where both the multiplicative conjunction and the lattice conjunction are included. In the \L{}ukasiewicz setting this distinction is essential, since $\wedge$ captures the order-theoretic meet, whereas $\odot$ captures the algebraic conjunction associated with the underlying t-norm.
Kripke semantics is generalized to the fuzzy context in a natural way. More precisely, we introduce the notion of ``fuzzy Kripke model'', a generalization of standard Kripke models where instead of relations between elements we have degrees of membership between them (see \cite{hansoul2013extending, jain2020fuzzy, vidal2021transitive, vidal2015modal}).
\begin{definition}[Fuzzy Models]
A \emph{fuzzy frame} is a tuple $\mathcal{F}=(W,R)$ where $W$ is a non empty set of states and $R:W^2\rightarrow[0,1]$ is a membership function indicating the membership degree of each edge.

A \emph{fuzzy Kripke model} is a triple $\mathbf{M}=(W,R,V)$ where $(W,R)$ is a fuzzy frame and $V:Prop\times W \rightarrow [0,1]$ is a the valuation function.  
\end{definition}
\begin{definition}[$\L$ukasiewicz truth evaluation]
\label{def:sem-luk}
Let $\mathbf{M}=(W,R,V)$ be a  fuzzy Kripke model. For every $w\in W$, the truth degree of a formula $\varphi$ at $w$ is defined recursively by a map $\truthsett{\,\cdot\,}{\mathbf{M}}{w}:Form\longrightarrow [0,1]$ as follows:
\begin{itemize}
    \item $\truthsett{p}{\mathbf{M}}{w}=V(p,w)$,

    \item $\truthsett{\bot}{\mathbf{M}}{w}=0$,

    \item $\truthsett{\varphi\,\odot\,\psi}{\mathbf{M}}{w}=\truthsett{\varphi}{\mathbf{M}}{w}\odot\truthsett{\psi}{\mathbf{M}}{w}$, \quad where $a\odot b=\max\{0,a+b-1\}$,

    \item $\truthsett{\varphi\wedge\psi}{\mathbf{M}}{w}=\min\left\{\truthsett{\varphi}{\mathbf{M}}{w},\truthsett{\psi}{\mathbf{M}}{w}\right\}$,

    \item $\truthsett{\varphi\rightarrow\psi}{\mathbf{M}}{w}=\truthsett{\varphi}{\mathbf{M}}{w}\Rightarrow_L\truthsett{\psi}{\mathbf{M}}{w}$, \quad where $a\Rightarrow_L b=\min\{1,1-a+b\}$, 

    \item $\truthsett{\Diamond\varphi}{\mathbf{M}}{w}=\sup_{w'\in W}\left(R(w,w')\odot\truthsett{\varphi}{\mathbf{M}}{w'}\right)$,

    \item$\truthsett{\Box\varphi}{\mathbf{M}}{w}=\inf_{w'\in W}\left(R(w,w')\Rightarrow_L\truthsett{\varphi}{\mathbf{M}}{w'}\right)$,
\end{itemize}
\end{definition}
The truth evaluation defined here is obtained by specializing the general many-valued Kripke semantics used in \cite{vidal2015modal} to the standard $\L$ukasiewicz algebra. In the general setting, modal formulas are interpreted over fuzzy Kripke models by using a chosen residuated lattice operation for conjunction and its corresponding residuum for implication.
Despite their flexibility, fuzzy Kripke semantics have some technical limitations. In particular, the normality axiom $K: \Box(\varphi\to\psi)\to(\Box\varphi\to\Box\psi)$ does not hold automatically over arbitrary many-valued Kripke frames, because the fuzzy accessibility relation interacts in a non-trivial way with many-valued implication (see \cite{vidal2015modal}). This challenge motivates the need for an alternative approach. Below we have provided an example to show that, the normality axiom is not valid in general in \L{}ukasiewicz fuzzy Kripke semantics
\begin{example}
Let $\mathbf{M}=(W,R,V)$ be a \L{}ukasiewicz fuzzy Kripke model, where
$W=\{w,v\}$. Define the fuzzy accessibility relation $R:W\times W\to[0,1]$ by
\[
    R(w,w)=1,\qquad R(w,v)=\frac{1}{2},\qquad
    R(v,w)=0,\qquad R(v,v)=1.
\]
Furthermore, define the valuation of the atomic propositions $p$ and $q$ by
\[
    V(p)(w)=1,\qquad V(q)(w)=1,
\]
and
\[
    V(p)(v)=\frac{1}{2},\qquad V(q)(v)=0.
\]
Recall that in \L{}ukasiewicz logic, $a\Rightarrow_{\L} b=\min(1,1-a+b)$.
The modal operator $\Box$ is interpreted by 
$\truthset{\Box\varphi}{\mathbf{M}}{}(w)=\inf_{u\in W}\bigl(R(w,u)\Rightarrow_{\L}\truthset{\varphi}{\mathbf{M}}{}(u)\bigr)$. We now evaluate the normality axiom $\Box(p\to q)\to(\Box p\to\Box q)$ at the world $w$. First, we compute the truth value of $p\to q$ at each world. At $w$, we have $\truthset{p\to q}{\mathbf{M}}{}(w)=1\Rightarrow_{\L}1=1$. At $v$, we have
$\truthset{p\to q}{\mathbf{M}}{}(v)=\frac{1}{2}\Rightarrow_{\L}0=\frac{1}{2}$. Hence, $\truthset{\Box(p\to q)}{\mathbf{M}}{}(w)=\inf_{u\in W}\bigl(R(w,u)\Rightarrow_{\L}\truthset{p\to q}{\mathbf{M}}{}(u)\bigr)=\min\{1\Rightarrow_{\L}1,\frac{1}{2}\Rightarrow_{\L}\frac{1}{2}\}=\min\{1,1\}=1$. Next, we compute $\Box p$ at $w$, $\truthset{\Box p}{\mathbf{M}}{}(w)=\inf_{u\in W}
    \bigl(R(w,u)\Rightarrow_{\L}\truthset{p}{\mathbf{M}}{}(u)\bigr)=\min\left\{1\Rightarrow_{\L}1,\frac{1}{2}\Rightarrow_{\L}\frac{1}{2}
    \right\}=\min\{1,1\}=1$.
Similarly, for $\Box q$ at $w$, we obtain $\truthset{\Box q}{\mathbf{M}}{}(w)=\frac{1}{2}$ Therefore,
\[
    \truthset{\Box p\to\Box q}{\mathbf{M}}{}(w)
    =
    1\Rightarrow_{\L}\frac{1}{2}
    =
    \frac{1}{2}.
\]
Finally, $\truthset{\Box(p\to q)\to(\Box p\to\Box q)}{\mathbf{M}}{}(w) =\frac{1}{2}$. Hence, the normality axiom is not valid in general in \L{}ukasiewicz fuzzy Kripke semantics.
\end{example}
\noindent\textbf{Fuzzy Topological Space:}
Fuzzy set theory was first introduced by Lotfi Zadeh in 1965 in \cite{zadeh1965fuzzy}, laying the foundation for fuzzy topology. Building on this, C.L. Chang introduced the concept of fuzzy topological spaces in 1968 in \cite{chang1968fuzzy}, defining  a fuzzy topology on a set $X$ as a family $\tau \subseteq I^X$ (with $I = [0,1]$) satisfying certain axioms, where each element of $\tau$ is regarded as a fuzzy open set. He also extended this framework to membership functions defined over arbitrary complete distributive lattices.
\begin{definition}[{\L}ukasiewicz Fuzzy operations]
    Let $X$ be a non-empty set. A fuzzy subset of $X$ is a function $\lambda: X\to [0,1]$. For fuzzy subsets $\lambda,\mu:X\to[0,1]$, we can define the pointwise operations:
\begin{align*}
    (\lambda\oplus\mu)(x)=\min\big(1,\lambda(x)+\mu(x)\big), \qquad&
    (\lambda\odot\mu)(x)=\max\big(0,\lambda(x)+\mu(x)-1\big),\\
    (\lambda\Rightarrow\mu)(x)=\min\big(1,1-\lambda(x)+\mu(x)\big), \qquad&
    (\lambda\land\mu)(x)=\min\big(\lambda(x),\mu(x)\big).
\end{align*}
\end{definition}
We write $0_{X}$ and $1_{X}$ for the constant fuzzy sets with values $0$ and $1$, respectively. The pointwise supremum is written as 
$(\lor_{i}\lambda_{i})(x)=\sup_{i}\lambda_{i}(x)$ and pointwise negation is $(\neg\lambda)(x)=1-\lambda(x)$. Before recalling the definition of a fuzzy topological space, we first fix the meaning of union and intersection for fuzzy sets. A fuzzy set $G$ on a nonempty set $X$ is identified with its membership function $G:X\to[0,1]$, where $G(x)$ represents the degree to which $x$ belongs to $G$. If $G_1$ and $G_2$ are fuzzy sets on $X$, then their intersection and union are defined pointwise by
$(G_1\cap G_2)(x)=\min\{G_1(x),G_2(x)\}$, and
$(G_1\cup G_2)(x)=\max\{G_1(x),G_2(x)\}$, for every $x\in X$. More generally, for an arbitrary family $\{G_j:j\in J\}$ of fuzzy sets, the fuzzy union is defined by
$\left(\bigcup_{j\in J}G_j\right)(x)=\sup_{j\in J}G_j(x)$.

\begin{definition}[Fuzzy Topological Space \cite{chang1968fuzzy}]
A family $\tau$ of fuzzy sets in $X$ is known as {\em fuzzy topology} on a nonempty set $X$ if it  satisfies the following axioms:
\begin{itemize}
    \item[(T1)] $0_{X}, 1_{X} \in \tau$
    \item[(T2)] $G_{1}\bigcap G_{2}\in \tau$, for any ${G_{1}},G_{2}\in \tau$
    \item[(T3)] $\bigcup G_{j}\in \tau$, for any arbitrary family $\{G_{j}:j\in J\}$, with $G_j\in\tau$.
\end{itemize}
The elements of $\tau$ are called a {\em fuzzy open set}.
\end{definition}
\begin{definition}[Fuzzy interior]
Let $(X,\tau)$ be a fuzzy topological space, then for a fuzzy subset
$\lambda:X \to [0,1]$, we define fuzzy interior as
$
\operatorname{Int}(\lambda)
=
\sup \{ \mu:\mu \in \tau, \mu \leq \lambda \}.
$
Thus $\operatorname{Int}(\lambda)$ is the greatest fuzzy open set contained in $\lambda$.
\end{definition}
\begin{definition}[Fuzzy closure]
Let $(X,\tau)$ be a fuzzy topological space, then for a fuzzy subset $\lambda:X \to [0,1]$, we define fuzzy closure as 
$
\operatorname{Cl}(\lambda)
=
\inf \{ \mu : \mu \text{ is closed, } \lambda \leq \mu \}.
$
It could also be defined as $ \operatorname{Cl}(\lambda)=\neg \operatorname{Int}(\neg \lambda)$. Equivalently,
$\operatorname{Cl}(\lambda)(x)=1-\operatorname{Int}(\neg \lambda)(x)$.
Hence interior and closure are connected by the duality
$\operatorname{Cl}(\lambda)=\neg \operatorname{Int}(\neg \lambda)$.
\end{definition}
A lattice is a partially ordered set $L$ in which every pair of elements in a meet and a join. If arbitrary meets and joins exist, then $L$ is
called a complete lattice. Goguen \cite{goguen1973fuzzy} generalized
Chang's fuzzy topology by replacing the unit interval $[0,1]$ with an
arbitrary lattice of truth values $L$. In this framework, fuzzy sets are
functions $A:X\to L$, and topological structures are defined directly in
terms of the algebraic operations of $L$ (see \cite{goguen1967fuzzy}). This approach is particularly
important for many-valued modal logics, since different logics are based on different truth-value algebras. It therefore provides a natural foundation for constructing a fuzzy topology compatible with the standard \L{}ukasiewicz algebra considered in this paper.

\section{{\L}ukasiewicz Modal Logic: topological semantics}\label{main section}
In this section, we introduce fuzzy topological semantics for modal logic over the standard {\L}ukasiewicz algebra. Building on the fuzzy topological approach previously introduced for the G\"odel case by Khan and Martins in \cite{khan2025fuzzytopologicalsemantics}, we now formulate the corresponding semantic framework for the {\L}ukasiewicz setting. The aim is to interpret modal formulas as fuzzy subsets of a Chang  fuzzy topological space \cite{chang1968fuzzy}, while evaluating the propositional connectives according to the operations of the standard {\L}ukasiewicz algebra. We define the associated models, give the recursive interpretation of formulas, and explain how the modal operators are represented through fuzzy topological constructions.

Definition \ref{language} contains several primitive connectives because, in the many-valued setting, the usual classical identification between different logical operations is no longer valid. In particular, the strong conjunction $\odot$, the residuated implication $\to$, and the lattice meet $\land$ represent
different operations in the underlying algebra. For instance, in \L ukasiewicz logic, $\odot$ is interpreted as a strong conjunction, while $\land$ is interpreted as the minimum operation, and these two operations do not coincide in general. Therefore, including these connectives makes the language rich enough to express both the algebraic structure of the fuzzy logic and the modal-topological principles studied later.

\begin{definition}[Fuzzy topological Model]
A fuzzy topological model is a triple $M=(X,\tau,V)$, where $(X,\tau)$ is a fuzzy topological space and $V:Prop\to [0,1]^X$ is a valuation ($V(p)(x)$ is the degree to which
$p$ is true at $x$).
\end{definition}
\begin{definition}
Let $M=(X,\tau,V)$ be a fuzzy topological model, for every formula $\varphi$, the interpretation $ \truthset{\varphi}{M}{} : X \longrightarrow [0,1] $ is defined recursively by
\begin{itemize}
    \item $\truthset{p}{M}{} = V(p)$,
    
    \item $\truthset{\bot}{M}{} = 0_X$,
    
    \item $\truthset{\varphi \odot \psi}{M}{}
    =
    \truthset{\varphi}{M}{} \odot \truthset{\psi}{M}{}$,
    
    \item $\truthset{\varphi \wedge \psi}{M}{}
    =
    \truthset{\varphi}{M}{} \wedge \truthset{\psi}{M}{}$,
    
    \item $\truthset{\varphi \to \psi}{M}{}
    =
    \truthset{\varphi}{M}{} \Rightarrow_{Ł} \truthset{\psi}{M}{}$,
    
    \item $\truthset{\Box \varphi}{M}{}
    =
    \operatorname{Int}\bigl(\truthset{\varphi}{M}{}\bigr)$,

    \item $ \truthset{\Diamond \varphi}{M}{} = \operatorname{Cl}\bigl(\truthset{\varphi}{M}{}\bigr)$.
\end{itemize}
\end{definition}
Since the connectives $\neg$, $\top$ and $\leftrightarrow$ are introduced as abbreviations in the language, their interpretations are obtained from the recursive clauses above. Hence, for every $\varphi,\psi\in Form$, $ \truthset{\neg\varphi}{M}{} = \truthset{\varphi\to\bot}{M}{} = \truthset{\varphi}{M}{} \Rightarrow_{L} 0_X$, $ \truthset{\top}{M}{} = \truthset{\bot\to\bot}{M}{} = 0_X \Rightarrow_{L} 0_X = 1_X$, and 
\[ 
\begin{aligned}
\truthset{\varphi\leftrightarrow\psi}{M}{} &= \truthset{(\varphi\to\psi)\wedge(\psi\to\varphi)}{M}{} \\ &= \bigl(\truthset{\varphi}{M}{} \Rightarrow_{L} \truthset{\psi}{M}{}\bigr) \wedge \bigl(\truthset{\psi}{M}{} \Rightarrow_{L} \truthset{\varphi}{M}{}\bigr). 
\end{aligned} 
\] 
Therefore, although these connectives are not taken as primitive in the recursive definition, their truth values are well defined through the primitive connectives of the language. The truth degree of $\varphi$ at $x\in X$ in $M$ is $\llbracket \varphi \rrbracket_M(x)$. We say that $\varphi$ is true at $x$ if $\llbracket \varphi \rrbracket_M(x)=1$. We say that $\varphi$ is valid in $M$ if $\llbracket \varphi \rrbracket_M=\mathbf{1}_X$,
equivalently, if $\llbracket \varphi \rrbracket_M(x)=1$ for every $x\in X$. A formula is valid in a fuzzy topological space $(X,\tau)$ if it is valid in every model based on $(X,\tau)$. Consequently, using the defining property of the \L{}ukasiewicz residuum (see \cite{hajek2001metamathematics}), we have that $\varphi\to\psi$ is valid in $M$ if and only if
$\truthset{\varphi}{M}{}\leq\truthset{\psi}{M}{}$. Indeed, the \L{}ukasiewicz implication satisfies $a\Rightarrow b=1$ if and only if $a\leq b$. This observation allows us to translate the validity of implication formulas into pointwise inequalities between fuzzy truth sets. Before studying the validity of the axioms of $S4$ in fuzzy topological models, we first develop the connection between \L{}ukasiewicz fuzzy Kripke frames and fuzzy topological models. This relationship provides the semantic foundation needed to transfer modal properties from the Kripke setting to the fuzzy topological setting.
\begin{definition}[Fuzzy Alexandroff Space \cite{bosifuzzy}]
Let $W$ be a set and $\tau \subseteq [0,1]^W$ a fuzzy topology on $W$. We say that $(W, \tau)$ is a \emph{fuzzy Alexandroff space} if,
$\tau$ is closed under arbitrary intersections (infima).
\end{definition}
To discuss the connection between fuzzy Kripke frames and fuzzy topological models, we focus on fuzzy Alexandroff topologies. Given a reflexive and Łukasiewicz-transitive fuzzy Kripke relation, we associate with it a family of (R)-upward fuzzy open sets and show that this family forms a fuzzy Alexandroff topology. We then establish that the interior and closure operators of the induced topology admit representations in terms of the corresponding fuzzy Kripke relation. As a consequence, the fuzzy Kripke model and its induced fuzzy topological model assign the same truth degree to every modal formula. Thus, fuzzy Alexandroff spaces provide a natural topological representation of reflexive and Łukasiewicz-transitive fuzzy Kripke semantics.
\begin{definition}[Reflexive and $\L$ukasiewicz Transitive Frame]
Let $\mathcal{F}=(W,R)$ be a fuzzy frame. We say that $\mathcal{F}$ is reflexive if 
\[ R(w,w)=1 \qquad \text{for every } w\in W. \] 
We say that $\mathcal{F}$ is $\L$ukasiewicz transitive if 
\[ R(w,v)\odot R(v,u)\leq R(w,u) \qquad \text{for all } w,v,u\in W. \] If both conditions hold, then $\mathcal{F}$ is called a reflexive and $\L$ukasiewicz transitive fuzzy Kripke frame.
\end{definition}
The second condition is the natural $\L$ukasiewicz analogue of fuzzy
transitivity. In the Gödel case, transitivity is expressed using the
minimum operation. In the $\L$ukasiewicz case, it is expressed using the strong conjunction $\odot$.
\begin{definition}[Induced Fuzzy Topology]
Given a reflexive and $\L$ukasiewicz transitive fuzzy Kripke frame $\mathcal{F}=(W,R)$, we define the family
\[
\tau_R
=
\left\{
\lambda\in[0,1]^W:
R(w,v)\odot \lambda(w)\leq \lambda(v)
\text{ for all }w,v\in W
\right\}.
\]
The elements of $\tau_R$ are called $R$-upward $\L$ukasiewicz fuzzy open sets.
\end{definition}
The condition $R(w,v)\odot \lambda(w)\leq \lambda(v)$ means that whenever $v$ is accessible from $w$ to degree $R(w,v)$, the
membership degree of $\lambda$ at $w$ is transferred to $v$ according to the $\L$ukasiewicz conjunction. By residuation, this condition is equivalent to $\lambda(w)\leq R(w,v)\Rightarrow_{\L}\lambda(v)$.

\begin{proposition}\label{prop:luk-induced-alexandroff}
For any reflexive and \L ukasiewicz transitive fuzzy Kripke frame $\mathcal{F}=(W,R)$ the induced family $\tau_R$ contains $0_W$ and $1_W$, and it is closed under arbitrary joins and arbitrary meets. Hence $(W,\tau_R)$ is a fuzzy Alexandroff topology.
\end{proposition}
\begin{proof}
We verify the required conditions, Let $1_W:W\to[0,1]$ be the constant fuzzy set with value $1$. For all
$w,v\in W$,
\[
R(w,v)\odot 1_W(w)
=
R(w,v)\odot 1
=
R(w,v)
\leq 1
=
1_W(v).
\]
Thus $1_W\in\tau_R$. Similarly, let $0_W:W\to[0,1]$ be the constant fuzzy set with value $0$.
Then $R(w,v)\odot 0_W(w)=R(w,v)\odot 0=0\leq 0=0_W(v)$. Thus $0_W\in\tau_R$.\\
Let $\{\lambda_i\}_{i\in I}\subseteq\tau_R$ and define $ \lambda=\bigvee_{i\in I}\lambda_i$, that is, $\lambda(w)=\sup_{i\in I}\lambda_i(w)$. For all $w,v\in W$, since $\odot$ preserves joins in each coordinate, we have
\[
R(w,v)\odot\lambda(w)
=
R(w,v)\odot\bigvee_{i\in I}\lambda_i(w)
=
\bigvee_{i\in I}\bigl(R(w,v)\odot\lambda_i(w)\bigr).
\]
Since each $\lambda_i\in\tau_R$, we have $R(w,v)\odot\lambda_i(w)\leq\lambda_i(v)$ for every $i\in I$. Therefore,
$\bigvee_{i\in I}\bigl(R(w,v)\odot\lambda_i(w)\bigr)\leq\bigvee_{i\in I}\lambda_i(v)
=\lambda(v)$. Hence $R(w,v)\odot\lambda(w)\leq\lambda(v)$, so $\lambda\in\tau_R$.\\
Let $\{\lambda_i\}_{i\in I}\subseteq\tau_R$ and define $ \mu=\bigwedge_{i\in I}\lambda_i$, that is, $\mu(w)=\inf_{i\in I}\lambda_i(w)$. Since $\mu(w)\leq\lambda_i(w)$ for every $i\in I$, and since $\odot$ is monotone, we have
$R(w,v)\odot\mu(w)\leq R(w,v)\odot\lambda_i(w)\leq\lambda_i(v)$ for every $i\in I$. Hence
\[
R(w,v)\odot\mu(w)
\leq
\bigwedge_{i\in I}\lambda_i(v)
=
\mu(v).
\]
Therefore $\mu\in\tau_R$. Thus $\tau_R$ contains $0_W$ and $1_W$, and it is closed under arbitrary joins and arbitrary meets. Hence $(W,\tau_R)$ is a fuzzy Alexandroff topology.
\end{proof}

For a fuzzy set $A:W\to[0,1]$, $\sup\{\mu(w):\mu\in\tau_R,\ \mu\leq A\}$ is the value at $w$ of the fuzzy interior of $A$, since it is obtained by taking the supremum of all fuzzy open sets below $A$. Similarly, $\inf\{\eta(w):\eta \text{ is closed},\ A\leq\eta\}$ is the value at $w$ of the fuzzy closure of $A$, since it is obtained by taking
the infimum of all fuzzy closed sets above $A$. The following proposition shows that, in the \L ukasiewicz fuzzy Alexandroff topology induced by $R$, these topological operations have the corresponding relational forms.

\begin{proposition}\label{prop:luk-int-cl-representation}
Let $(W,\tau_R)$ be the fuzzy Alexandroff topology induced by a
reflexive and $\L$ukasiewicz transitive fuzzy relation $ R:W\times W\to[0,1]$. Then, for every fuzzy set $A:W\to[0,1]$ and every $w\in W$, the following hold:
\begin{enumerate}
    \item 
$
    \sup\{\mu(w):\mu\in\tau_R,\ \mu\leq A\}
    =
    \inf_{v\in W}
    \bigl(R(w,v)\Rightarrow_{\L}A(v)\bigr).
$
\item
$
    \inf\{\eta(w):\eta \text{ is closed},\ A\leq\eta\}
    =
    \sup_{v\in W}
    \bigl(R(w,v)\odot A(v)\bigr).
$
\end{enumerate}
\end{proposition}
\begin{proof}
We prove both parts separately.
\begin{enumerate}
\item
Define a fuzzy set $B:W\to[0,1]$ by $ B(w)=\inf_{v\in W}
\bigl(R(w,v)\Rightarrow_{\L}A(v)\bigr)$. We show that $B$ is the largest open fuzzy set below $A$. First, by reflexivity, $R(w,w)=1$. Therefore,
\[
B(w)
\leq
R(w,w)\Rightarrow_{\L}A(w)
=
1\Rightarrow_{\L}A(w)
=
A(w).
\]
Hence $B\leq A$. Now to show that $B\in\tau_R$, we must prove that $R(w,v)\odot B(w)\leq B(v)$ for all $w,v\in W$. By the definition of $B(v)$, it is enough to prove that, for every
$u\in W$,
\[
R(w,v)\odot B(w)
\leq
R(v,u)\Rightarrow_{\L}A(u).
\]
By residuation, this is equivalent to proving $ R(w,v)\odot B(w)\odot R(v,u)\leq A(u)$. Using associativity and commutativity of $\odot$, together with $\L$ukasiewicz transitivity, we get
\[
R(w,v)\odot R(v,u)\odot B(w)
\leq
R(w,u)\odot B(w).
\]
By the definition of $B(w)$, $B(w)\leq R(w,u)\Rightarrow_{\L}A(u)$. Applying residuation again, we obtain $R(w,u)\odot B(w)\leq A(u)$. Therefore, $R(w,v)\odot B(w)\odot R(v,u)\leq A(u)$. Hence $R(w,v)\odot B(w)\leq R(v,u)\Rightarrow_{\L}A(u)$ for every $u\in W$. Taking the infimum over $u$, we obtain $R(w,v)\odot B(w)\leq B(v)$.\\
Thus $B\in\tau_R$. Now let $\mu\in\tau_R$ with $\mu\leq A$. Since $\mu\in\tau_R$, we have $R(w,v)\odot\mu(w)\leq\mu(v)$ for every $v\in W$. Since $\mu\leq A$, $ R(w,v)\odot\mu(w)\leq A(v)$. By residuation, $
\mu(w)\leq R(w,v)\Rightarrow_{\L}A(v)$ for every $v\in W$. Therefore,
\[
\mu(w)
\leq
\inf_{v\in W}
\bigl(R(w,v)\Rightarrow_{\L}A(v)\bigr)
=
B(w).
\]
This shows that every open fuzzy set below $A$ is below $B$, and we have
already shown that $B$ itself is open and below $A$. Hence
\[
\sup\{\mu(w):\mu\in\tau_R,\ \mu\leq A\}
=
B(w)
=
\inf_{v\in W}
\bigl(R(w,v)\Rightarrow_{\L}A(v)\bigr).
\]

\item
Define a fuzzy set $C:W\to[0,1]$ by $C(w)=\sup_{v\in W} \bigl(R(w,v)\odot A(v)\bigr)$. We show that $C$ is the smallest closed fuzzy set above $A$. First, by reflexivity, $R(w,w)=1$. Hence
\[
C(w)
\geq
R(w,w)\odot A(w)
=
1\odot A(w)
=
A(w).
\]
Thus $A\leq C$. Now we show that $C$ is closed. We need to prove that $
R(w,v)\odot C(v)\leq C(w)$ for all $w,v\in W$. By definition,
$C(v)=\sup_{u\in W}\bigl(R(v,u)\odot A(u)\bigr)$. Since $\odot$ preserves joins in each coordinate,
\[
R(w,v)\odot C(v)
=
\sup_{u\in W}
\bigl(R(w,v)\odot R(v,u)\odot A(u)\bigr).
\]
By $\L$ukasiewicz transitivity, $R(w,v)\odot R(v,u)\leq R(w,u)$. Therefore, $R(w,v)\odot R(v,u)\odot A(u)\leq R(w,u)\odot A(u)$. Taking the supremum over all $u\in W$, we obtain
\[
R(w,v)\odot C(v)
\leq
\sup_{u\in W}
\bigl(R(w,u)\odot A(u)\bigr)
=
C(w).
\]
Thus $C$ is closed and contains $A$. Now let $\eta$ be any closed fuzzy set such that $A\leq\eta$. Since $\eta$
is closed, $R(w,v)\odot\eta(v)\leq\eta(w)$. Also, since $A(v)\leq\eta(v)$, monotonicity of $\odot$ gives
$R(w,v)\odot A(v)\leq R(w,v)\odot\eta(v)\leq\eta(w)$. Taking the supremum over all $v\in W$, we obtain
\[
C(w)
=
\sup_{v\in W}
\bigl(R(w,v)\odot A(v)\bigr)
\leq
\eta(w).
\]
Thus $C$ is below every closed fuzzy set containing $A$. Since $C$ itself is closed and contains $A$, it is the smallest closed fuzzy set above $A$. Therefore,
\[
\inf\{\eta(w):\eta \text{ is closed},\ A\leq\eta\}
=
C(w)
=
\sup_{v\in W}
\bigl(R(w,v)\odot A(v)\bigr).
\]
\end{enumerate}
\end{proof}

The previous proposition provides the main link between the relational and topological viewpoints. Indeed, the fuzzy interior and fuzzy closure induced by the fuzzy relation $R$ have exactly the same pointwise form as the modal operators $\Box$ and $\Diamond$ in the corresponding fuzzy Kripke semantics. In what follows, for a given fuzzy Kripke model $\mathbf{M}^{K}=(W,R,V)$, we consider the associated fuzzy topological model $M^{T}=(W,\tau_R,V)$, where $\tau_R$ is the fuzzy topology determined by $R$. We then show that every modal formula receives the same truth-value in both models.

\begin{proposition}\label{prop:luk-kripke-topological-equivalence}
Let $\mathcal{F}=(W,R)$ be a reflexive and \L ukasiewicz transitive fuzzy Kripke frame, and let $(W,\tau_R)$ be the corresponding
fuzzy Alexandroff space. For any fuzzy valuation $V$ and any
modal formula $\varphi$, we have
\[
\truthsett{\varphi}{\mathbf{M}^K}{w}
=
\truthset{\varphi}{M^T}{w},
\]
where
\[
\mathbf{M}^K=(W,R,V)
\quad\text{and}\quad
M^T=(W,\tau_R,V).
\]
\end{proposition}

\begin{proof}
The proof is by induction on the structure of $\varphi$. For propositional variables, the result is immediate because both models
use the same valuation $V$. The cases for the propositional connectives are also immediate, since both semantics use the same $\L$ukasiewicz algebraic operations pointwise. Now suppose that the result holds for a formula $\psi$. We prove it for $\square\psi$. By the definition of the $\L$ukasiewicz fuzzy Kripke model,
\[
\truthsett{\square\psi}{\mathbf{M}^K}{w}
=
\inf_{v\in W}
\bigl(
R(w,v)\Rightarrow_{\L}
\truthsett{\psi}{\mathbf{M}^K}{v}
\bigr).
\]
By the induction hypothesis, $ \truthsett{\psi}{\mathbf{M}^K}{v} =
\truthset{\psi}{M^T}{v} $ for every $v\in W$. Therefore,
\[
\truthsett{\square\psi}{\mathbf{M}^K}{w}
=
\inf_{v\in W}
\bigl(
R(w,v)\Rightarrow_{\L}
\truthset{\psi}{M^T}{v}
\bigr).
\]
By Proposition~\ref{prop:luk-int-cl-representation}, this is equal to
$ Int_{\L}\bigl(\truthset{\psi}{M^T}{}\bigr)(w)$. By the definition of the topological semantics,
\[
Int_{\L}
\bigl(\truthset{\psi}{M^T}{}\bigr)(w)
=
\truthset{\square\psi}{M^T}{w}.
\]
Hence
\[
\truthsett{\square\psi}{\mathbf{M}^K}{w}
=
\truthset{\square\psi}{M^T}{w}.
\]
The proof for $\diamondsuit\psi$ is similar. Thus, by induction, the equality holds for every modal formula $\varphi$.
\end{proof}
\textbf{Soundness of $S4$}: Before proving the validity of $S4$ in the fuzzy topological setting, we must address a structural issue unique to the {\L}ukasiewicz case, which does not arise in the G\"odel setting described in \cite{khan2025fuzzytopologicalsemantics}. In Chang's original fuzzy topology (\cite{chang1968fuzzy}), open sets are closed under arbitrary joins and finite lattice meets. For the standard G\"odel algebra, the monoidal conjunction coincides with the lattice meet, making topological closure automatic. However, in the standard {\L}ukasiewicz algebra, the strong conjunction $a \odot b=\max(0,a+b-1)$ differs from the lattice meet $a\wedge b=\min(a,b)$. Consequently, closure under finite meets does not guarantee closure under the strong conjunction. This distinction is critical when proving the normal modal axiom $\Box(\varphi\to\psi)\to(\Box\varphi\to\Box\psi)$. Interpreting $\Box$ as a fuzzy interior operator requires demonstrating that for arbitrary fuzzy sets $A,B\in [0,1]^X$,
\[
\operatorname{Int}(A\to B) \leq \operatorname{Int}(A)\to \operatorname{Int}(B)
\]
By residuation in the {\L}ukasiewicz algebra, $\operatorname{Int}(A\to B)\odot\operatorname{Int}(A)\leq(A\to B)\odot A\leq B$. To conclude that $\operatorname{Int}(A\to B)\odot\operatorname{Int}(A)\leq\operatorname{Int}(B)$, the fuzzy set $\operatorname{Int}(A\to B)\odot \operatorname{Int}(A)$ must itself be open. Chang's topology, lacking closure under the monoidal operation $\odot$, cannot guarantee this, preventing the general proof of the $K$ axiom. Strengthening Chang's topology to include closure under $\odot$ is therefore structurally necessary to ensure compatibility with the monoidal conjunction of {\L}ukasiewicz logic. A parallel phenomenon exists in many-valued Kripke semantics, validating the $K$ axiom over residuated lattices often requires semantic restrictions, such as idempotent or crisp accessibility relations (see \cite{vidal2015modal, hansoul2013extending}). Rather than restricting the accessibility relation, our approach strengthens the fuzzy topology to ensure the family of open fuzzy sets is closed under strong conjunction. The additional axiom,
\[
\lambda,\mu\in\tau_L \quad\Longrightarrow\quad \lambda\odot\mu\in\tau_L
\]
ensures the interior operator behaves normally with respect to the algebra's residuated structure. This validates the $K$ axiom, while $T$ and $4$ follow from standard interior properties, establishing a rigorous semantic foundation for $S4$ over {\L}ukasiewicz fuzzy topological models. To define {\L}ukasiewicz fuzzy topological spaces, we synthesize two established frameworks. We adopt closure under arbitrary joins and a monoidal operation from Goguen’s $L$-topology in \cite{goguen1973fuzzy}, and retain closure under finite lattice meets from classical (see \cite{kelley1955general}) and Chang's fuzzy topology (see \cite{chang1968fuzzy}). By taking $L=[0,1]$ equipped with {\L}ukasiewicz operations and setting $\odot$ as the monoidal operation, we define the {\L}ukasiewicz fuzzy topology as follows,
\begin{definition}
    A {\L}ukasiewicz fuzzy topology on $X$ is a family $\tau_{L}\subseteq[0,1]^X$ satisfying:
  
  $
    \begin{array}{l}
        0_X,1_X \in \tau_L,\\
        \lambda,\mu\in\tau_L \Rightarrow \lambda\land\mu\in\tau_L,\\
        \lambda,\mu\in\tau_L \Rightarrow \lambda \odot\mu\in\tau_L,\\
        \{\lambda_i\}_{i\in I}\subseteq\tau_L \Rightarrow \lor_{i\in I} \lambda_i\in\tau_L.
    \end{array}
    $
\end{definition}
The pair $(X,\tau_L)$ is is called a {\L}ukasiewicz fuzzy topological space, and the members of $\tau_L$ are called {\L}ukasiewicz fuzzy open sets. This definition specializes the general lattice-valued topological idea of open fuzzy sets being closed under arbitrary joins and a monoidal operation, while also retaining finite lattice meets needed for the $S4$ meet principle.

\begin{proposition}[Interior and closure properties]\label{prop:int-closure}
Let $(X,\tau_L)$ be a {\L}ukasiewicz fuzzy topological space. For all fuzzy subsets $\lambda,\mu : X \to [0,1]$, the following properties hold:
\begin{center}
\begin{tabular}{p{0.47\textwidth} p{0.47\textwidth}}
\begin{minipage}[t]{0.47\textwidth}
\begin{enumerate}
    \item $\operatorname{Int}_{L}(1_X) = 1_X$,

    \item $\operatorname{Int}_{L}(\lambda) \leq \lambda$,

    \item $\operatorname{Int}_{L}\bigl(\operatorname{Int}_{L}(\lambda)\bigr)
    =
    \operatorname{Int}_{L}(\lambda)$,

    \item $\operatorname{Int}_{L}(\lambda \wedge \mu)
    =
    \operatorname{Int}_{L}(\lambda) \wedge \operatorname{Int}_{L}(\mu)$,

    \item $\operatorname{Int}_{L}(\lambda) \odot \operatorname{Int}_{L}(\mu)
    \leq
    \operatorname{Int}_{L}(\lambda \odot \mu)$.
\end{enumerate}
\end{minipage}
&
\begin{minipage}[t]{0.47\textwidth}
\begin{enumerate}
    \item $\operatorname{Cl}_{L}(0_X) = 0_X$,

    \item $\lambda \leq \operatorname{Cl}_{L}(\lambda)$,

    \item $\operatorname{Cl}_{L}\bigl(\operatorname{Cl}_{L}(\lambda)\bigr)
    =
    \operatorname{Cl}_{L}(\lambda)$,

    \item $\operatorname{Cl}_{L}(\lambda\vee\mu)=\operatorname{Cl}_{L}(\lambda)\vee\operatorname{Cl}_{L}(\mu)$.
\end{enumerate}
\end{minipage}

\end{tabular}
\end{center}
\end{proposition}
\begin{proof}
The first standard properties of $\operatorname{Int}_{L}$ and $\operatorname{Cl}_{L}$ follow directly from the usual definitions of fuzzy interior and fuzzy closure in a fuzzy topological space (see \cite{lowen1976fuzzy}). Hence, we only prove the additional property involving the {\L}ukasiewicz strong conjunction. Since $\operatorname{Int}_{L}(\lambda)\leq \lambda$
and $\operatorname{Int}_{L}(\mu)\leq \mu$, and since the \L{}ukasiewicz strong conjunction $\odot$ is monotone in both arguments, we get $\operatorname{Int}_{L}(\lambda)\odot\operatorname{Int}_{L}(\mu)\leq\lambda\odot\mu$. Also, $\operatorname{Int}_{L}(\lambda),\operatorname{Int}_{L}(\mu)\in\tau_L$, and $\tau_L$ is closed under $\odot$. Therefore,
$\operatorname{Int}_{L}(\lambda)\odot \operatorname{Int}_{L}(\mu)\in\tau_L$. So $\operatorname{Int}_{L}(\lambda)\odot \operatorname{Int}_{L}(\mu)$ is a \L{}ukasiewicz fuzzy open set contained in $\lambda\odot\mu$. By definition of the interior as the greatest open fuzzy set below $\lambda\odot\mu$, we obtain
\[
\operatorname{Int}_{L}(\lambda)\odot \operatorname{Int}_{L}(\mu)
\leq
\operatorname{Int}_{L}(\lambda\odot\mu).
\]
\end{proof}
The interior properties can be naturally viewed as modal principles when the fuzzy interior operator $\operatorname{Int}$ is replaced by the modal operator $\Box$. Under this interpretation, they correspond closely to the standard axioms of the modal system $S4$. In the same way, the closure properties, when formulated in terms of the diamond operator $\diamondsuit$ instead of the fuzzy closure operator Cl, reflect the dual form of the $S4$ principles. As we show below, each fuzzy topological space therefore determines a semantic model in which the axioms of $S4$ are valid.

\begin{theorem}
Let $\mathsf{LFTop}$ be the class of all \L{}ukasiewicz fuzzy topological spaces. Then the axioms of S4 are valid in every \L{}ukasiewicz fuzzy topological model over $\mathsf{LFTop}$.
\end{theorem}
\begin{proof}
Let $M=(X,\tau_L,V)$ be an arbitrary \L{}ukasiewicz fuzzy topological model and let $\varphi,\psi\in Form$. All inequalities between fuzzy subsets are understood pointwise. Hence, to show that a formula $\chi$ is valid in $M$, it is enough to show that $\truthset{\chi}{M}{} = 1_X$.

The validity of the principles $(N)$, $(T)$, $(4)$, the meet-preservation
principle $(R)$, Modus Ponens, and the monotonicity rule for $\Box$ follows from the general properties of the fuzzy interior operator. These principles were already established for arbitrary fuzzy topological models in our previous work on fuzzy topological semantics \cite{khan2025fuzzytopologicalsemantics}.
Indeed, they depend only on the standard properties
\[
    \operatorname{Int}(1_X)=1_X,\qquad
    \operatorname{Int}(\lambda)\leq \lambda,\qquad
    \operatorname{Int}(\operatorname{Int}(\lambda))
    =
    \operatorname{Int}(\lambda),
\]
together with the preservation of finite meets and the monotonicity of $\operatorname{Int}$. Therefore, we do not repeat these arguments here. In the present \L{}ukasiewicz setting, the only point requiring a separate verification is the normality axiom $ \Box(\varphi\to\psi)\to(\Box\varphi\to\Box\psi)$. This is because normality involves the interaction between the
\L{}ukasiewicz residuated implication, the fuzzy interior operator, and the
strong conjunction $\odot$. Hence, we now prove the validity of the normality
axiom.\\
Normality principle $\Box(\varphi \to \psi) \to (\Box\varphi \to \Box\psi)$: By {\L}ukasiewicz logic, $A \to (B \to C)$
has value $1$ if $A\leq B\to C$ which means $A \odot B \leq C$.
Therefore, to prove normality, it is enough to show,
\[
\llbracket \Box(\varphi \to \psi) \rrbracket_M
\odot
\llbracket \Box\varphi \rrbracket_M
\leq
\llbracket \Box\psi \rrbracket_M.
\]
To show validity of this expression, we need to verify the following inequality,
\[
\llbracket \Box(\varphi \to \psi) \rrbracket_M
\odot
\llbracket \Box\varphi \rrbracket_M
\leq
\llbracket \Box\psi \rrbracket_M.
\]
By definition, 
$
\llbracket \Box(\varphi \to \psi) \rrbracket_M
=
\operatorname{Int}_{L}(\llbracket \varphi \to \psi \rrbracket_M),
$ 
and $
\llbracket \Box\varphi \rrbracket_M
=
\operatorname{Int}_{L}(\llbracket \varphi \rrbracket_M).
$ 
Both are open fuzzy sets because they are interiors. Since $\tau_L$ is closed under $\odot$, 
$
\llbracket \Box(\varphi \to \psi) \rrbracket_M
\odot
\llbracket \Box\varphi \rrbracket_M
\in \tau_L.
$ 
So, the left-hand side is open. Also, we know:
\[
\operatorname{Int}_{L}(\alpha) \leq \alpha.
\]
Therefore, 
$
\llbracket \Box(\varphi \to \psi) \rrbracket_M
\leq
\llbracket \varphi \to \psi \rrbracket_M,
$ 
and 
$
\llbracket \Box\varphi \rrbracket_M
\leq
\llbracket \varphi \rrbracket_M.
$ 
By monotonicity of $\odot$,
\[
\llbracket \Box(\varphi \to \psi) \rrbracket_M
\odot
\llbracket \Box\varphi \rrbracket_M
\leq
\llbracket \varphi \to \psi \rrbracket_M
\odot
\llbracket \varphi \rrbracket_M.
\]
Pointwise,
\[
\llbracket \varphi \to \psi \rrbracket_M
=
\llbracket \varphi \rrbracket_M
\Rightarrow
\llbracket \psi \rrbracket_M.
\]
The residuation inequality gives, 
$
(a \Rightarrow b) \odot a \leq b.
$ 
So, 
$
\llbracket \varphi \to \psi \rrbracket_M
\odot
\llbracket \varphi \rrbracket_M
\leq
\llbracket \psi \rrbracket_M.
$ 
Therefore, 
$
\llbracket \Box(\varphi \to \psi) \rrbracket_M
\odot
\llbracket \Box\varphi \rrbracket_M
\leq
\llbracket \psi \rrbracket_M.
$ 
Since $\operatorname{Int}_{L}(\llbracket \psi \rrbracket_M)$ is the greatest open fuzzy set below $\llbracket \psi \rrbracket_M$, we get
\[
\llbracket \Box(\varphi \to \psi) \rrbracket_M
\odot
\llbracket \Box\varphi \rrbracket_M
\leq
\operatorname{Int}_{L}(\llbracket \psi \rrbracket_M).
\]
But $\operatorname{Int}_{L}(\llbracket \psi \rrbracket_M)=\llbracket \Box\psi \rrbracket_M$. Hence,
$\llbracket \Box(\varphi \to \psi) \rrbracket_M\odot\llbracket\Box\varphi \rrbracket_M\leq\llbracket \Box\psi \rrbracket_M$. Therefore,
\[
\Box(\varphi \to \psi) \to (\Box\varphi \to \Box\psi)
\].
\end{proof}
\begin{example}
Let $X=\{x,y\}$ and consider the following family of fuzzy subsets of $X$:
\[
\tau=\{(0,0), (0.6,0), (0.8,1), (1,1)\}.
\]
Here, a fuzzy set $A$ on $X$ is written as $A=(A(x),A(y))$. Clearly, $0_X=(0,0)$ and $1_X=(1,1)$ belong to $\tau$. Moreover, since $\tau$ is a chain under the pointwise order, it is closed under finite lattice meets and arbitrary joins. Hence, $\tau$ is a Chang's fuzzy topology on $X$ (see \cite{chang1968fuzzy}). However, $\tau$ is not closed under the {\L}ukasiewicz strong conjunction $\odot$. Now let $M=(X,\tau,V)$ be a fuzzy topological model, where $V(p)=(0.6,0)$ and $V(q)=(0.4,0)$. Since $V(p)=(0.6,0)$ is already open, we have $\truthset{\Box p}{M}{}=\operatorname{Int}(V(p))=(0.6,0)$. On the other hand, the only open fuzzy set below $V(q)=(0.4,0)$ is $(0,0)$. Therefore, $\truthset{\Box q}{M}{}=\operatorname{Int}(V(q))=(0,0)$.
One can easily compute the implication $p\rightarrow q$, 
\[
\truthset{p\to q}{M}{} = (0.8,1).
\]
Since $(0.8,1)\in\tau$, it is already open, and therefore $\truthset{\Box(p\to q)}{M}{} = \operatorname{Int}(\truthset{p\to q}{M}{}) = (0.8,1)$. We now evaluate the $K$-formula
\[
\truthset{\Box(p\to q)\to(\Box p\to\Box q)}{M}{} = (0.6,1)\neq(1,1).
\]
Therefore, the $K$-formula is not valid in $M$. This shows that if the fuzzy topology is not closed under the {\L}ukasiewicz strong conjunction $\odot$, then the normality axiom
\[
\Box(p\to q)\to(\Box p\to\Box q)
\]
cannot be proved in general.
\end{example}

\section{ Conclusion and Future Work}
In this paper, we developed a Łukasiewicz fuzzy topological semantics for fuzzy modal logic. It constructs on joining the algebraic structure of the standard Łukasiewicz algebra to a fuzzy topological reading of the modal operators. Within this setting a formula denotes a fuzzy subset of a space, and the operators $\Box$ and $\Diamond$ act as fuzzy interior and closure operators, respectively. The governing topology is required to be compatible with the strong conjunction $\odot$, so that the non-idempotent behaviour of the Łukasiewicz truth operations is faithfully preserved. The standard $S4$ principles were verified for the resulting class of models. Alongside them, we obtained a compatibility principle particular to Łukasiewicz logic that governs the strong conjunction.\\
These results form part of a wider effort to build fuzzy topological semantics for many-valued modal logics. Because each such logic rests on its own truth-value algebra, the accompanying topological structure has to be tailored to the algebraic operations of its semantics. While our previous work was focused on the Gödel case, the present paper addresses the Łukasiewicz case. Several directions remain open. Completeness, bisimulations, and model equivalence for Łukasiewicz fuzzy topological models all worth further study.

\textbf{Acknowledgments.} {\footnotesize The present study was developed in the scope of the Project ``Agenda ILLIANCE'' [C644919832-00000035 | Project n° 46], financed by PRR - Plano de Recuperação e Resiliência under the Next Generation EU from the European Union, by CIDMA under the Portuguese Foundation for Science and Technology (FCT) Multi-Annual Financing Program for R$\&$D Units, grants UID/4106/2025 and UID/PRR/4106/ 2025 and by the Project Bansky, with reference compete2020-feder-00892000.}

\nocite{*}
\bibliographystyle{eptcs}
\bibliography{references}
\end{document}